\documentclass[runningheads,a4paper]{llncs}

\usepackage{amsmath}
\usepackage{amsfonts}
\usepackage{amssymb}

\usepackage{graphicx}

\usepackage{pstricks}
\usepackage{pst-node}
\usepackage{pst-eucl}

\usepackage{algorithm}
\usepackage{algpseudocode}

\newcommand{\integ}{\mathbb{Z}}
\newcommand{\real}{\mathbb{R}}

\DeclareMathOperator{\conv}{conv}

\allowdisplaybreaks

\begin{document}

\title{Fast Convex Decomposition for Truthful Social Welfare Approximation}

\author{Dennis Kraft \and Salman Fadaei \and Martin Bichler}
\authorrunning{Kraft, Fadaei and Bichler}
\tocauthor{Dennis Kraft, Salman Fadaei, Martin Bichler}

\institute{Department of Informatics, TU M\"unchen, Munich, Germany \\ 
\email{dennis.kraft@in.tum.de, salman.fadaei@in.tum.de, bichler@in.tum.de}}
\maketitle


\begin{abstract}

Approximating the optimal social welfare while preserving truthfulness is a well studied problem in algorithmic mechanism design. Assuming that the social welfare of a given mechanism design problem can be optimized by an integer program whose integrality gap is at most $\alpha$, Lavi and Swamy~\cite{Lavi11} propose a general approach to designing a randomized $\alpha$-approximation mechanism which is truthful in expectation. Their method is based on decomposing an optimal solution for the relaxed linear program into a convex combination of integer solutions. Unfortunately, Lavi and Swamy's decomposition technique relies heavily on the ellipsoid method, which is notorious for its poor practical performance. To overcome this problem, we present an alternative decomposition technique which yields an $\alpha(1 + \epsilon)$ approximation and only requires a quadratic number of calls to an integrality gap verifier.


\end{abstract}


\section{Introduction}\label{sec:Introduction}

Optimizing the social welfare in the presence of self-interested players poses two main challenges to algorithmic mechanism design. On the one hand, the social welfare consists of the player's valuations for possible outcomes of the mechanism. However, since these valuations are private information, they can be misrepresented for personal advantage. To avoid strategic manipulation, which may harm the social welfare, it is important to encourage truthful participation. In mechanism design, this is achieved through additional payments which offer each player a monetary incentive to reveal his true valuation. Assuming that the mechanism returns an optimal outcome with respect to the reported valuations, the well known Vickrey, Clarke and Groves (VCG) principle~\cite{vickrey61,Clarke71,Groves73} provides a general method to design payments such that each player maximizes his utility if he reports his valuation truthfully. On the other hand, even if the player's valuations are known, optimizing the social welfare is NP-hard for many combinatorial mechanism design problems. Since an exact optimization is intractable under these circumstances, the use of approximation algorithms becomes necessary. Unfortunately, VCG payments are generally not compatible with approximation algorithms.


To preserve truthfulness, so called maximal-in-range (MIR) approximation algorithms must be used~\cite{nisan2000computationally}. This means there must exist a fixed subset of outcomes, such that the approximation algorithm performs optimally with respect to this subset. Given that the players are risk-neutral, the concept of MIR algorithms can be generalized to distributions over outcomes. Together with VCG payments, these maximal-in-distribution-range (MIDR) algorithms allow for the design of randomized approximation mechanisms such that each player maximizes his expected utility if he reveals his true valuation~\cite{Dobzinski09}. This property, which is slightly weaker than truthfulness in a deterministic sense, is also referred to as truthfulness in expectation. 


A well-known method to convert general approximation algorithms which verify an integrality gap of $\alpha$ into MIDR algorithms is the linear programing approach of Lavi and Swamy~\cite{Lavi11}. Conceptually, their method is based on the observation that scaling down a packing polytope by its integrality gap yields a new polytope which is completely contained in the convex hull of the original polytope's integer points. Considering that the social welfare of many combinatorial mechanism design problems can be expressed naturally as an integer program, this scaled polytope corresponds to a set of distributions over the outcomes of the mechanism. Thus, by decomposing a scaled solution of the relaxed linear program into a convex combination of integer solutions, Lavi and Swamy obtain an $\alpha$-approximation mechanism which is MIDR.


Algorithmically, Lavi and Swamy's work builds on a decomposition technique by Carr and Vempala~\cite{carr2000randomized}, which uses a linear program to decompose the scaled relaxed solution. However, since this linear program might have an exponential number of variables, one for every outcome of the mechanism, it can not be solved directly. Instead, Carr and Vempala use the ellipsoid method in combination with an integrality gap verifier to identify a more practical, but still sufficient, subset of outcomes for the decomposition. Although this approach only requires a polynomial number of calls to the integrality gap verifier in theory, the ellipsoid method is notoriously inefficient in practice~\cite{bland1981ellipsoid}. 


In this work, we propose an alternative decomposition technique which does not rely on the ellipsoid method but is general enough to substitute Carr and Vempala's~\cite{carr2000randomized} decomposition technique. The main component of our decomposition technique is an algorithm which computes a convex combination within an arbitrarily small distance $\epsilon$ to the scaled relaxed solution. However, since an exact decomposition is necessary to guarantee truthfulness, we slightly increase the scaling factor of the relaxed solution and apply a post-processing step to match our convex combination with the additionally scaled relaxed solution. Assuming that $\epsilon$ is positive and fixed, our technique yields an $\alpha(1 + \epsilon)$ approximation of the optimal social welfare but uses only a quadratic number of calls to the integrality gap verifier.

\newpage


\section{Setting}\label{sec:Setting}

Integer programming is a powerful tool in combinatorial optimization. Using binary variables to indicate whether certain goods are allocated to a player, the outcomes of various NP-hard mechanism design problems, such as combinatorial auctions or generalized assignment problems~\cite{Lavi11,Dughmi2010a}, can be modeled as integer points of an $n$-dimensional packing polytope $X \subseteq [0,1]^n$.

\begin{definition}\textbf{(Packing Polytope)}
	Polytope $X$ satisfies the packing property if all points $y$ which are dominated by some point $x$ from $X$ are also contained in $X$
	
	\[\forall x, y \in \real_{\geq 0}^n:~ x \in X \wedge x \geq y \Rightarrow y \in X.\]
\end{definition}
Together with a vector $\mu \in \real_{\geq 0}^n$ which denotes the accumulated valuations of the players, it is possible to express the social welfare as an integer program of the form $\max_{x \in \integ(X)} \sum_{k = 1}^n \mu_k x_k$, where $\integ(X)$ denotes the set of integer points in $X$. Clearly, the task of optimizing the social welfare remains NP-hard, regardless of its representation. Nevertheless, an optimal solution $x^* \in X$ for the relaxed linear program $\max_{x \in X} \sum_{k = 1}^n \mu_k x_k$ can be computed in polynomial time.


The maximum ratio between the original program and its relaxation is called the integrality gap of $X$. Assuming this gap is at most $\alpha \in \real_{\geq1}$, Lavi and Swamy~\cite{Lavi11} observe that the scaled fractional solution $\frac{x^*}{\alpha}$ can be decomposed into a convex combination of integer solutions. More formally, there exists a convex combination $\lambda$ from the set $\Lambda = \{\lambda \in \real_{\geq_0}^{\integ(X)} \mid \sum_{x \in \integ(X)} \lambda_x = 1\}$ such that the point $\sigma(\lambda)$, which is defined as $\sigma(\lambda) = \sum_{x \in \integ(X)} \lambda_x x$, is equal to $\frac{x^*}{\alpha}$. Regarding $\lambda$ as a probability distribution over the feasible integer solutions, the MIDR principle allows for the construction of a randomized $\alpha$-approximation mechanism which is truthful in expectation.


From an algorithmic point of view, the decomposition of $\frac{x^*}{\alpha}$ requires to compute several integer points in $X$. Unfortunately, the number of these points might be exponential with respect to $n$, which makes it intractable to consider the entire set $\integ(X)$. However, not all integer points in $\integ(X)$ are necessarily needed for a successful decomposition. For instance, given an approximation algorithm $\mathcal{A} : \real_{\geq 0}^n \rightarrow \integ(X)$ which verifies an integrality gap of $\alpha$, Carr and Vempala~\cite{carr2000randomized} propose a decomposition technique which computes a suitable and sufficient subset of integer points based on a polynomial number of calls to $\mathcal{A}$.

\begin{definition}\textbf{(Integrality Gap Verifier)}
	Approximation algorithm $\mathcal{A}$ verifies an integrality gap of $\alpha$ if the integer solution which is computed by $\mathcal{A}$ is at least $\alpha$ times the optimal relaxed solution for all non-negative vectors $\mu$
	
	\[\forall \mu \in \real_{\geq 0}^n:~ \alpha\sum_{k=1}^n\mu_k\mathcal{A}(\mu)_k \geq \max_{x \in X}\sum_{k=1}^n\mu_k x_k.\]
\end{definition}
In particular, this implies that the number of positive coefficients in the resulting decomposition $\lambda$, which is denoted by $\psi(\lambda) = |\{x \in \integ(X) \mid \lambda_x > 0\}|$, is polynomial as well. Nevertheless, considering that Carr and Vempala's approach strongly relies on the ellipsoid method, it is clear that this decomposition technique is more of theoretical importance than of practical use.


\section{Decomposition with Epsilon Precision}\label{sec:DecompositionwithEpsilonPrecision}

The first part of our decomposition technique is to construct a convex combination $\lambda$ such that the point $\sigma(\lambda)$ is within an arbitrarily small distance $\epsilon \in \real_{> 0}$ to the scaled relaxed solution $\frac{x^*}{\alpha}$. Similar to Carr and Vempala's approach, our technique requires an approximation algorithm $\mathcal{A}' : \real^n \rightarrow \integ(X)$ to sample integer points from $X$. It is important to note that $\mathcal{A}'$ must verify an integrality gap of $\alpha$ for arbitrary vectors $\mu \in \real^n$ whereas $\mathcal{A}$, only accepts non-negative vectors. However, since $X$ satisfies the packing property, it is easy to extend the domain of $\mathcal{A}$ while preserving an approximation ratio of $\alpha$.

\begin{lemma}
\label{ExtVer}
	Approximation algorithm $\mathcal{A}$ can be extended to a new approximation algorithm $\mathcal{A}'$ which verifies an integrality gap of $\alpha$ for arbitrary vectors $\mu$. 
\end{lemma}

\begin{proof}
	The basic idea of $\mathcal{A}'$ is to replace all negative components of $\mu$ by $0$ and run the original integrality gap verifier $\mathcal{A}$ on the resulting non-negative vector, which is defined as $\xi(\mu)_k = \max(\{\mu_k, 0\})$. Exploiting the fact that $X$ is a packing polytope, the output of $\mathcal{A}$ is then set to $0$ for all negative components of $\mu$. More formally, $\mathcal{A}'$ is defined as
	
	\[\mathcal{A}'(\mu)_k  = \begin{cases}
    			\mathcal{A}(\xi(\mu))_k & \text{if } \mu_k \geq 0 \\
   			0       & \text{if } \mu_k < 0.
		\end{cases}\]
		

	Since $\mathcal{A}'(\mu)_k$ is equal to $0$ if $\mu_k$ is negative and otherwise corresponds to $\mathcal{A}(\xi(\mu))_k$, it holds that
	
	\[\sum_{k=1}^n \mu_k \mathcal{A}'(\mu)_k = \sum_{k=1}^n \xi(\mu)_k \mathcal{A}'(\mu)_k = \sum_{k=1}^n \xi(\mu)_k \mathcal{A}(\xi(\mu))_k.\]
	Furthermore, since $X$ only contains non-negative points, $\max_{x \in X}\sum_{k=1}^n \xi(\mu)_k x_k$ must be greater or equal to $\max_{x \in X}\sum_{k=1}^n \mu_k x_k$. Together with the fact that $\mathcal{A}$ verifies an integrality gap of $\alpha$ for $\xi(\mu)$ this proves that $\mathcal{A}'$ verifies the same integrality gap for $\mu$
	\[\alpha \sum_{k=1}^n \mu_k \mathcal{A}'(\mu)_k = \alpha \sum_{k=1}^n \xi(\mu)_k \mathcal{A}(\xi(\mu))_k \geq \max_{x \in X}\sum_{k=1}^n \xi(\mu)_k x_k \geq \max_{x \in X}\sum_{k=1}^n \mu_k x_k.\]
	\qed
\end{proof}


Once $\mathcal{A}'$ is specified, algorithm~\ref{alg:EpsilonDecomposition} is used to decompose $\frac{x^*}{\alpha}$. Starting at the origin, which can be expressed trivially as a convex combination from $\Lambda$ due to the packing property of $X$, the algorithm gradually improves $\sigma(\lambda^i)$ until it is sufficiently close to $\frac{x^*}{\alpha}$. For each iteration of the algorithm, $\mu^i$ denotes the vector which points from $\sigma(\lambda^i)$ to $\frac{x^*}{\alpha}$. If the length of $\mu^i$ is less or equal to $\epsilon$, then $\sigma(\lambda^i)$ must be within an $\epsilon$-distance to $\frac{x^*}{\alpha}$ and the algorithm terminates. Otherwise, $\mathcal{A}'$ samples a new integer point $x^{i + 1}$ based on the direction of $\mu^i$. It is important to observe that all points on the line segment between $\sigma(\lambda^i)$ and $x^{i + 1}$ can be expressed as a convex combination of the form $\delta \lambda^i + (1 - \delta)\tau(x^{i + 1})$, where $\delta$ is a value between $0$ and $1$ and $\tau(x^{i + 1})$ denotes a convex combination such that the coefficient $\tau(x^{i + 1})_{x^{i + 1}}$ is equal to $1$ while all other coefficients are $0$. Thus, by choosing $\lambda^{i + 1}$ as the convex combination which minimizes the distance between the line segment and $\frac{x^*}{\alpha}$, an improvement of the current decomposition may be possible. In fact, theorem~\ref{thm:IterationBound} states that at most $\lceil n \epsilon^{-2}\rceil -1$ iterations are necessary to achieve the desired precision of $\epsilon$.

\begin{algorithm}
\caption{Decomposition with Epsilon Precision}
\label{alg:EpsilonDecomposition}
	\begin{algorithmic}
		\Require an optimal relaxed solution~$x^*$, an approximation algorithm $\mathcal{A}'$, a precision $\epsilon$
		\Ensure a convex combination $\lambda$ which is within an $\epsilon$-distance to $\frac{x^*}{\alpha}$
		\State $x^0 \gets 0,~ \lambda^0 \gets \tau(x^0),~ \mu^0 \gets \frac{x^*}{\alpha} - \sigma(\lambda^0),~ i \gets 0$
		\While{$\|\mu^i\|_2 > \epsilon$}
			\State $x^{i + 1} \gets \mathcal{A}'(\mu^i)$			
			\State $\delta \gets \operatorname*{arg\,min}_{\delta \in [0,1]} \| \frac{x^*}{\alpha} - (\delta \sigma(\lambda^i) + (1 - \delta) x^{i + 1})\|_2$
			\State $\lambda^{i + 1} \gets \delta \lambda^i + (1 - \delta) \tau(x^{i + 1})$
			\State $\mu^{i + 1} \gets \frac{x^*}{\alpha} - \sigma(\lambda^{i + 1})$
			\State $i \gets i + 1$
		\EndWhile
		\State \Return $\lambda^i$
	\end{algorithmic}
\end{algorithm}

\begin{theorem}
\label{thm:IterationBound}
	Algorithm~\ref{alg:EpsilonDecomposition} returns a convex combination within an $\epsilon$-distance to the scaled relaxed solution $\frac{x^*}{\alpha}$ after at most $\lceil n \epsilon^{-2}\rceil -1$ iterations. 
\end{theorem}

\begin{proof}
	Clearly, algorithm~\ref{alg:EpsilonDecomposition} terminates if and only if the distance between $\sigma(\lambda^i)$ and $\frac{x^*}{\alpha}$ becomes less or equal to $\epsilon$. Thus, suppose the length of vector $\mu^i$ is still greater than $\epsilon$. Consequently, approximation algorithm $\mathcal{A}'$ is deployed to sample a new integer point $x^{i + 1}$. Keeping in mind that $\mathcal{A}'$ verifies an integrality gap of $\alpha$, the value of $x^{i + 1}$ must be greater or equal to the value of $\frac{x^*}{\alpha}$ with respect to vector $\mu^i$ 
	
	\[\sum_{k=1}^n \mu_k^i x_k^{i + 1} = \sum_{k=1}^n \mu_k^i \mathcal{A}'(\mu^i)_k \geq \max_{x \in X}\sum_{k=1}^n \mu_k^i \frac{x_k}{\alpha} \geq \sum_{k=1}^n \mu_k^i \frac{x^*}{\alpha}.\]
Conversely, since the squared distance between $\sigma(\lambda^i)$ and $\frac{x^*}{\alpha}$ is greater than $\epsilon^2$, and therefore also greater than $0$, it holds that the value of $\sigma(\mu^i)$ is less than the value of $\frac{x^*}{\alpha}$ with respect to vector $\mu^i$

	\begin{align*}
						& & 0	& < \sum_{k=1}^n\Big(\frac{x_k^*}{\alpha} - \sigma(\lambda^i)_k\Big)^2\\
		\Longleftrightarrow	& & 0	& < \sum_{k=1}^n\Big(\Big(\frac{x_k^*}{\alpha}\Big)^2 - 2\frac{x_k^*}{\alpha}\sigma(\lambda^i)_k + \sigma(\lambda^i)_k^2\Big)\\
		\Longleftrightarrow	& & \sum_{k=1}^n\Big(\frac{x_k^*}{\alpha}\sigma(\lambda^i)_k - \sigma(\lambda^i)_k^2\Big) & < \sum_{k=1}^n\Big(\Big(\frac{x_k^*}{\alpha}\Big)^2 - \frac{x_k^*}{\alpha}\sigma(\lambda^i)_k\Big) \\
		\Longleftrightarrow	& & \sum_{k=1}^n \mu_k^i \sigma(\lambda^i)_k & < \sum_{k=1}^n\mu_k^i \frac{x_k^*}{\alpha}.
	\end{align*}
As a result, the hyper plane $\{x \in \real^n \mid \sum_{k=1}^n \mu_k^i x_k = \sum_{k=1}^n \mu_k^i \frac{x_k^*}{\alpha}\}$ separates $\sigma(\lambda^i)$ from $x^{i + 1}$, which in turn implies that the line segment $\conv(\{\sigma(\lambda^i), x^{i + 1}\})$ intersects the hyperplane at a unique point $z^{i + 1}$.

	
	\begin{figure}[t]
		\begin{center}
			\begin{pspicture}(-5.25, -3.25)(5.25, 1.75)
				\psset{xunit=1cm, yunit=1cm, runit=1cm}
				\psset{arrowsize=0.15, arrowinset=0.0,arrowlength=1.5}
 	 
				\pnode(-5.25,0){hl}
				\pnode(5.25,0){hr}
				\cnode*(4,0){.1}{fs}
				\uput{.25}[45](4,0){$\frac{x^*}{\alpha}$}
				\cnode*(4,-3){.1}{l}
				\uput{.25}[-135](4,-3){$\sigma(\lambda^i)$}
				\cnode*(2.56,-1.92){.1}{nl}
				\uput{.25}[-135](2.56,-1.92){$\sigma(\lambda^{i + 1})$}
				\cnode*(0,0){.1}{z}
				\uput{.25}[-135](0,0){$z^{i + 1}$}
				\cnode*(-2,1.5){.1}{x}
				\uput{.25}[-135](-2,1.5){$x^{i + 1}$}

				\uput{.0}[135](3.28,-0.96){$\mu^{i+1}$}
				\uput{.1}[0](4,-1.5){$\mu^{i}$}
 
				\ncline[linestyle=dashed]{hl}{z}
				\ncline[linestyle=dashed]{fs}{hr}
				\uput{.15}[45](-5.25,0){hyperplane}

				\ncline{->}{l}{fs}
				\ncline{->}{nl}{fs}

				\ncline{l}{z}
				\ncline[linestyle=dashed]{z}{x}

				\ncline{z}{fs}

				\pstRightAngle{z}{fs}{l}
				\pstRightAngle{z}{nl}{fs}
			\end{pspicture}
		\end{center}
		\caption{Right triangle between the points $\frac{x^*}{\alpha}$, $\sigma(\lambda_x^i)$ and $z^{i + 1}$}
		\label{fig:triangle}
	\end{figure}
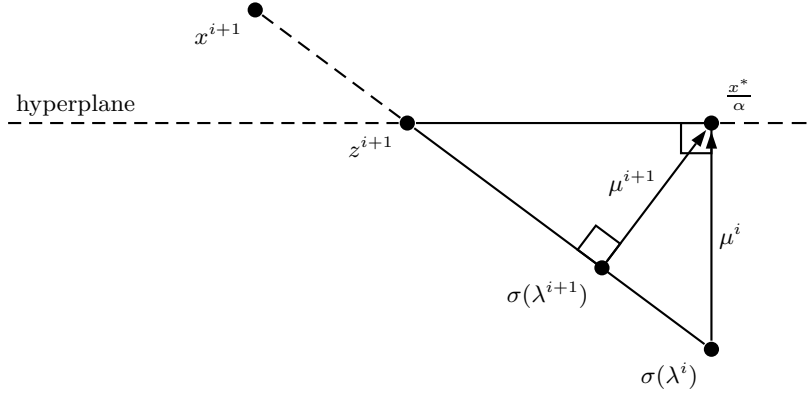
	
	Since the hyperplane is orthogonal to $\mu^i$, the points $\frac{x^*}{\alpha}$, $\sigma(\lambda_x^i)$ and $z^{i + 1}$ form a right triangle, as figure \ref{fig:triangle} illustrates. Furthermore, the altitude of this triangle minimizes the distance from the line segment $\conv(\{\sigma(\lambda^i), x^{i+1}\})$ to $\frac{x^*}{\alpha}$ and therefore corresponds to the length of new vector $\mu^{i + 1}$. According to the basic relations between the sides in a right triangle, the length of $\mu^{i + 1}$ can be expressed as
	
	\[\big\|\mu^{i + 1}\big\|_2 = \sqrt{\frac{\big\|\mu^i\big\|_2^2 \big\|\frac{x^*}{\alpha} - z^{i+1}\big\|_2^2}{\big\|\mu^i\big\|_2^2 + \big\|\frac{x^*}{\alpha} - z^{i+1}\big\|_2^2}}.\]
	
	
	Unfortunately, the exact position of $z^{i + 1}$, depends on the implementation $\mathcal{A}'$. To obtain an upper bound on the length $\mu^{i + 1}$ which does not rely on $z^{i + 1}$, it is helpful to observe that the altitude of the triangle grows as the distance between $z^{i + 1}$ and $\frac{x^*}{\alpha}$ increases. However, since both points are contained in the standard hyper cube $[0,1]^n$, the square of this distance is at most $n$
	
	\[\big\|\frac{x^*}{\alpha} - z^{i+1}\big\|_2^2 = \sum_{k=1}^n\Big(\frac{x_k^*}{\alpha} - z_k^{i+1}\Big)^2 \leq \sum_{k=1}^n1 = n,\]
	which means that the maximum length of $\mu^{i + 1}$ is given by
	
	\begin{align*}
						& & \big\|\frac{x^*}{\alpha} - z^{i+1}\big\|_2^2 																			& \leq n\\
		\Longleftrightarrow	& & \frac{\big\|\frac{x^*}{\alpha} - z^{i+1}\big\|_2^2}{\big\|\mu^i\big\|_2^2 + \big\|\frac{x^*}{\alpha} - z^{i+1}\big\|_2^2}						& \leq \frac{n}{\big\|\mu^i\big\|_2^2 + n}\\
		\Longleftrightarrow	& & \frac{\big\|\mu^i\big\|_2^2\big\|\frac{x^*}{\alpha} - z^{i+1}\big\|_2^2}{\big\|\mu^i\big\|_2^2 + \big\|\frac{x^*}{\alpha} - z^{i+1}\big\|_2^2}		& \leq \frac{\big\|\mu^i\big\|_2^2n}{\big\|\mu^i\big\|_2^2 + n}\\
		\Longleftrightarrow	& & \sqrt{\frac{\big\|\mu^i\big\|_2^2\big\|\frac{x^*}{\alpha} - z^{i+1}\big\|_2^2}{\big\|\mu^i\big\|_2^2 + \big\|\frac{x^*}{\alpha} - z^{i+1}\big\|_2^2}}	& \leq \sqrt{\frac{\big\|\mu^i\big\|_2^2n}{\big\|\mu^i\big\|_2^2 + n}}.\\
	\end{align*}
	
	
	\begin{figure}[t]
		\begin{center}
			\begin{pspicture}(-4.50, -4.50)(4.50, 4.50)
				\psset{xunit=0.85cm, yunit=0.85cm, runit=0.85cm}
				\psset{arrowsize=0.15, arrowinset=0.0,arrowlength=1.5}
	 
		 		\pscircle[linestyle=dashed](0,0){5}
		 
				\cnode*(0,0){.1}{fs}
				\uput{.25}[-45](0,0){$\frac{x^*}{\alpha}$}
			
				\cnode*(0, -5){.1}{l0}
				\uput{.25}[-90](0, -5){$\sigma(\lambda^0)$}
				\cnode*(-2.5, -2.5){.1}{l1}
				\uput{.25}[-135](-2.5, -2.5){$\sigma(\lambda^1)$}
				\cnode*(-2.845177969, -0.488155365){.1}{l2}
				\uput{.25}[-170.265](-2.845177969, -0.488155365){$\sigma(\lambda^2)$}
				\cnode*(-2.345260950, 0.865881676){.1}{l3}
				\uput{.25}[159.736](-2.345260950, 0.865881676){$\sigma(\lambda^3)$}
				\cnode*(-1.529856090, 1.630809720){.1}{l4}
				\uput{.25}[133.171](-1.529856090, 1.630809720){$\sigma(\lambda^4)$}
				\cnode*(-0.667113176, 1.929151802){.1}{l5}
				\uput{.25}[109.076](-0.667113176, 1.929151802){$\sigma(\lambda^5)$}
				\cnode*(0.103251213, 1.886999670){.1}{l6}
				\uput{.25}[86.8681](0.103251213, 1.886999670){$\sigma(\lambda^6)$}
				\cnode*(0.714411293, 1.616977582){.1}{l7}
				\uput{.25}[66.1633](0.714411293, 1.616977582){$\sigma(\lambda^7)$}
			
				\cnode*(-5, 0){.1}{z1}
				\uput{.25}[180](-5 ,0){$z^1$}
				\cnode*(-3.535533906, 3.535533906){.1}{z2}
				\uput{.25}[135](-3.535533906, 3.535533906){$z^2$}
				\cnode*(-0.845509894, 4.927992798){.1}{z3}
				\uput{.25}[99.7356](-0.845509894, 4.927992798){$z^3$}
				\cnode*(1.731763352, 4.690521900){.1}{z4}
				\uput{.25}[69.7356](1.731763352, 4.690521900){$z^4$}
				\cnode*(3.646601393, 3.420862213){.1}{z5}
				\uput{.25}[43.1706](3.646601393, 3.420862213){$z^5$}
				\cnode*(4.725437552, 1.634086882){.1}{z6}
				\uput{.25}[19.0757](4.725437552, 1.634086882){$z^6$}
				\cnode*(4.992531854, -0.273177032){.1}{z7}
				\uput{.25}[-3.13194](4.992531854, -0.273177032){$z^7$}

				\ncline{fs}{l0}
				\naput{$\sqrt{n}$}
				\ncline{l0}{l1}			
				\ncline[linestyle=dashed]{l1}{z1}
				\ncline{fs}{l1}
				\ncline{l1}{l2}
				\ncline[linestyle=dashed]{l2}{z2}
    				\ncline{fs}{l2}
				\ncline{l2}{l3}
				\ncline[linestyle=dashed]{l3}{z3}
				\ncline{fs}{l3}
				\ncline{l3}{l4}
				\ncline[linestyle=dashed]{l4}{z4}
				\ncline{fs}{l4}
				\ncline{l4}{l5}
				\ncline[linestyle=dashed]{l5}{z5}
				\ncline{fs}{l5}
				\ncline{l5}{l6}
				\ncline[linestyle=dashed]{l6}{z6}
				\ncline{fs}{l6}
				\ncline{l6}{l7}
				\ncline[linestyle=dashed]{l7}{z7}
				\ncline{fs}{l7}
			\end{pspicture}
		\end{center}
		\caption{Upper bound on the distance between $\sigma(\lambda^i)$ and $\frac{x^*}{\alpha}$ for the first $7$ iterations}
		\label{fig:convergence}
	\end{figure}
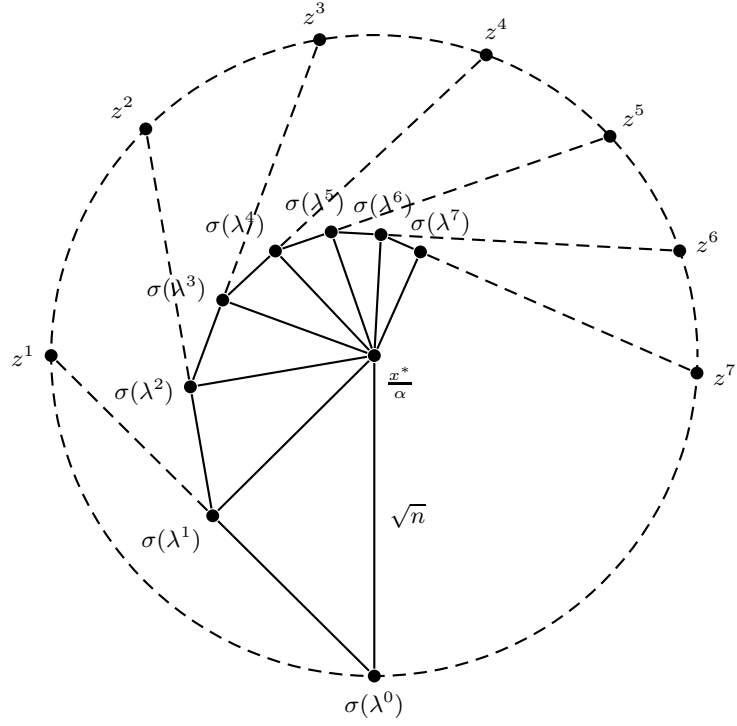
	
	It is important to note that this upper bound on the length of $\mu^{i + 1}$, which is illustrated in figure~\ref{fig:convergence}, only depends on the previous vector $\mu^i$ and the number of dimensions $n$. Solving the recurrence inequality yields yet another upper bound which is based on the initial vector $\mu^0$ and the number of iterations $i$
	
	\begin{align*}
						& & \big\|\mu^i\big\|_2^2 & \leq \frac{\big\|\mu^{i-1}\big\|_2^2 n}{\big\|\mu^{i-1}\big\|_2^2 + n}\\
	 	\Longleftrightarrow	& & \frac{\big\|\mu^i\big\|_2^2}{n} & \leq \frac{\big\|\mu^{i-1}\big\|_2^2}{\big\|\mu^{i-1}\big\|_2^2 + n}\\
	 	\Longleftrightarrow	& & \frac{n}{\big\|\mu^i\big\|_2^2} & \geq \frac{n}{\big\|\mu^{i -1}\big\|_2^2} + 1~...\\
		\Longrightarrow	& & \frac{n}{\big\|\mu^i\big\|_2^2} & \geq \frac{n}{\big\|\mu^0\big\|_2^2} + i\\
		\Longleftrightarrow	& & \frac{\big\|\mu^i\big\|_2^2}{n} & \leq \frac{\big\|\mu^0\big\|_2^2}{\big\|\mu^0\big\|_2^2i + n}\\
		\Longleftrightarrow	& & \big\|\mu^i\big\|_2 & \leq \sqrt{\frac{\big\|\mu^{0}\big\|_2^2 n}{\big\|\mu^{0}\big\|_2^2i + n}}.\\
	\end{align*}
	

	Considering that the squared length of vector $\mu^0$, which corresponds to the distance between $\frac{x^*}{\alpha}$ and the origin, is at most $n$
	
	\[\big\|\mu^0\big\|_2^2 = \sum_{k=1}^n \Big(\frac{x_k^*}{\alpha}\Big)^2 \leq \sum_{k=1}^n1 = n,\]
	it follows that
	
	\[\big\|\mu^i\big\|_2 \leq \sqrt{\frac{\big\|\mu^{0}\big\|_2^2 n}{\big\|\mu^{0}\big\|_2^2i + n}} \leq \sqrt{\frac{n^2}{ni + n}} = \sqrt{\frac{n}{i + 1}}.\]
	Finally, this proves that the distance between $\sigma(\lambda^i)$ and $\frac{x^*}{\alpha}$ must be less or equal to $\epsilon$ after not more than $\lceil n \epsilon^{-2}\rceil -1$ iterations, at which point the algorithm terminates
	
	\[\big\|\mu^{\lceil n \epsilon^{-2}\rceil-1}\big\|_2 \leq \sqrt{\frac{n}{1 + (\lceil n \epsilon^{-2}\rceil-1)}} \leq \epsilon.\]
	\qed
\end{proof}


It should be mentioned that according to theorem~\ref{thm:IterationBound}, the total number of iterations is linear with respect to $n$. Since each of these iterations adds at most one additional point to $\lambda$, the number of positive coefficients $\psi(\lambda)$ must be linear as well. Considering that the decomposition of a fractional point in $n$-dimensional space requires up to $n+1$ affinely independent integer points, it follows that for any fixed $\epsilon$ the performance of our decomposition algorithm is asymptotically optimal.


\section{Exact Decomposition}\label{sec:ExactDecomposition}

Although the convex combination $\lambda$ which is returned by algorithm \ref{alg:EpsilonDecomposition} is within an $\epsilon$-distance to $\frac{x^*}{\alpha}$, an exact decomposition of the relaxed solution is necessary to guarantee truthfulness. Assuming that an additional scaling factor of $\sqrt{n}\epsilon$ is admissible, the second part of our decomposition technique shows how to convert $\lambda$ into a new convex combination $\lambda''$ such that $\sigma(\lambda'')$ is equal to $\frac{x^*}{\alpha(1+\sqrt{n}\epsilon)}$. It is important to note that this additional scaling factor depends on $\epsilon$, which implies that it can still be made arbitrarily small. In particular, running algorithm \ref{alg:EpsilonDecomposition} with a precision of $\frac{\epsilon}{\sqrt{n}}$, instead of $\epsilon$, reduces the factor to $\epsilon$ and yields a decomposition which is equal to~$\frac{x^*}{\alpha(1+\epsilon)}$. However, since this new precision is not independent of $n$ anymore, the maximum iteration number is increased to $\lceil n (\frac{\epsilon}{\sqrt{n}})^{-2}\rceil - 1$, which is quadratic in $n$.


To adjust $\sigma(\lambda)$ component-wisely, it is helpful to consider the integer points $e^k \in \{0,1\}^n$. For every dimension $k$, the $k$th component of $e^k$ is defined to be $1$ while all other components are $0$. Since $X$ has a finite integrality gap and also satisfies the packing property, all points $e^k$ must be contained in $X$.

\begin{lemma}
\label{thm:StandardBasisPoints}
	The polytope $X$ contains all points $e^k$.
\end{lemma}

\begin{proof}
	For the sake of contradiction, assume there exists a dimension $k$ for which $e^k$ is not contained in $X$. Since $X$ satisfies the packing property, this implies that there exists no point in $X$ whose $k$th component is $1$, in particular no integer point. As a result, the optimal solution for the integer program with respect to the vector $e^k$ must be $0$
	
	\[\max_{x \in \integ(X)}  \sum_l^n {e_l^k} x_l = \max_{x \in \integ(X)} x_k = 0.\]
	Keeping in mind that $X$ has an integrality gap of at most $\alpha$, it immediately follows that the optimal solution for the relaxed linear program with respect to $e^k$ must also be $0$
	
	\[\max_{x \in X}  \sum_l^n {e_l^k} x_l = \max_{x \in X} x_k = 0.\]
	However, this implies that the $k$th component of every point in $X$ is $0$, which contradicts the fact that $X$ is $n$-dimensional.
	\qed
\end{proof}
Applying theorem~\ref{thm:DominatingConvexCombination}, our decomposition technique uses the points $e^k$ to construct an intermediate convex combination $\lambda'$ such that $\sigma(\lambda')$ dominates $\frac{x^*}{\alpha(1+\sqrt{n}\epsilon)}$.

\begin{theorem}
\label{thm:DominatingConvexCombination}
	Convex combination $\lambda$ can be converted into a new convex combination $\lambda'$ which dominates $\frac{x^*}{\alpha(1+\sqrt{n}\epsilon)}$.
\end{theorem}

\begin{proof}
	According to lemma~\ref{thm:StandardBasisPoints}, the points $e^k$ are contained in $\integ(X)$. Thus, they can be added to $\lambda$ to construct a positive combination $\lambda + \sum_{k=1}^n |\frac{x_k^*}{\alpha} - \sigma(\lambda)_k| \tau(e^k)$ which dominates $\frac{x^*}{\alpha}$
	
	\begin{eqnarray*}
		\sigma\Big(\lambda + \sum_{k=1}^n\Big| \frac{x_k^*}{\alpha} - \sigma(\lambda)_k\Big| \tau(e^k)\Big)	& = & \sigma(\lambda) + \Big(\sum_{k=1}^n\Big| \frac{x_k^*}{\alpha} - \sigma(\lambda)_k\Big| e^k\Big)\\
																					& \geq & \sigma(\lambda) + \Big(\sum_{k=1}^n \Big(\frac{x_k^*}{\alpha} - \sigma(\lambda)_k\Big) e^k\Big)\\
																					& = & \sigma(\lambda) + \frac{x^*}{\alpha} - \sigma(\lambda)\\
																					& = & \frac{x^*}{\alpha}.
	\end{eqnarray*}
	
	
	Since the sum over the additional coefficients $\sum_{k=1}^n|\frac{x_k^*}{\alpha} - \sigma(\lambda)_k|$ is equivalent to the L1 distance between $\sigma(\lambda)$ and $\frac{x^*}{\alpha}$, it can be upper bounded by the H\"older inequality
		
	\[\sum_{k=1}^n\Big|\frac{x_k^*}{\alpha} - \sigma(\lambda)_k\Big| = \Big\|\frac{x^*}{\alpha} - \sigma(\lambda)\Big\|_1 \leq \Big\|1\Big\|_2 \Big\|\frac{x^*}{\alpha} - \sigma(\lambda)\Big\|_2 \leq \sqrt{n} \epsilon.\]
	As a result, scaling down the positive combination by a factor of $1 + \sqrt{n} \epsilon$ yields a new positive combination which dominates $\frac{x^*}{\alpha(1+\sqrt{n}\epsilon)}$ and whose coefficients sum up to a value less or equal to $1$. To ensure that this sum becomes exactly $1$, the coefficients must be increased by an additional value of $\sqrt{n}\epsilon - \sum_{k=1}^n|\frac{x_k^*}{\alpha} - \sigma(\lambda)_k|$. An easy way to achieve this is by adding the origin, which is trivially contained in $\integ(X)$ due to the packing property of $X$, to the positive combination. Thus, the desired convex combination $\lambda'$ corresponds to
	
	\[\frac{\lambda + \sum_{k=1}^n\big| \frac{x_k^*}{\alpha} - \sigma(\lambda)_k\big| \tau(e^k) + \big(\sqrt{n}\epsilon - \sum_{k=1}^n\big|\frac{x_k^*}{\alpha} - \sigma(\lambda)_k\big|\big) \tau(0)}{1 + \sqrt{n}\epsilon}.\]
	\qed
\end{proof}


In the final step, our decomposition technique exploits the packing property of $X$ to convert $\lambda'$ into an exact decomposition of $\frac{x^*}{\alpha(1+\sqrt{n}\epsilon)}$. A simple but general approach to this problem is provided by algorithm~\ref{alg:DominatingToExactDecomposition}. Given a point $x \in X$ which is dominated by $\sigma(\lambda')$, the basic idea of the algorithm is to iteratively weaken the integer points which comprise $\lambda'$ until the desired convex combination $\lambda''$ is reached. As theorem~\ref{thm:DominatingToExact} shows, this computation requires at most $|\psi(\lambda)|n + \frac{n^2+n}{2}$ iterations.

\begin{algorithm}
\caption{From a Dominating to an Exact Decomposition}
\label{alg:DominatingToExactDecomposition}
	\begin{algorithmic}
		\Require a convex combination $\lambda'$,  a point $x$ which is dominated by $\sigma(\lambda')$
		\Ensure a convex combination $\lambda''$ which is an exact decomposition of $x$
		\State $\lambda^0 \gets \lambda',~ i \gets 0$
		\ForAll{$1 \leq k \leq n$}
			\While{$\sigma(\lambda^i)_k > x_k$}
				\State $y \gets \text{ pick some $y$ from $\integ(X)$ such that } \lambda_y^i > 0 \text{ and } y_k = 1$
				\If {$\lambda^i_y \geq \sigma(\lambda^i)_k - x_k$}
					\State $\lambda^{i + 1} \gets \lambda^i - (\sigma(\lambda^i)_k - x_k) \tau(y) + (\sigma(\lambda^i)_k - x_k) \tau(y - e^k)$
				\Else
					\State $\lambda^{i + 1} \gets \lambda^i - \lambda^i_y \tau(y) + \lambda^i_y \tau(y - e^k)$
				\EndIf
				\State $i \gets i + 1$
			\EndWhile
		\EndFor
		\State \Return $\lambda^i$
	\end{algorithmic}
\end{algorithm}

\begin{theorem}
\label{thm:DominatingToExact}
	Assuming that $\sigma(\lambda')$ dominates the point $x$, algorithm~\ref{alg:DominatingToExactDecomposition} converts $\lambda'$ into a new convex combination $\lambda''$ such that $\sigma(\lambda'')$ is equal to $x$. Furthermore, the required number of iterations is at most $|\psi(\lambda')|n + \frac{n^2+n}{2}$.
\end{theorem}

\begin{proof}
	In order to match $\sigma(\lambda')$ with $x$, algorithm \ref{alg:DominatingToExactDecomposition} considers each dimension $k$ separately. Clearly, while $\sigma(\lambda^i)_k$ is still greater than $x_k$, there must exist at least one point $y$ in $\lambda^i$ which has a value of $1$ in component $k$. If $\lambda_y^i$ is greater or equal to the difference between $\sigma(\lambda^i)_k$ and $x_k$, it is reduced by the value of this difference. To compensate for this operation, the coefficient of the point $y - e^k$, which is trivially contained in $X$ due to its packing property, is increased by the same value. Thus, the value of $\sigma(\lambda^{i + 1})_k$ is equal to $x_k$
	
	\begin{eqnarray*}
		\sigma(\lambda^{i + 1})_k	& = & \sigma(\lambda^i)_k - (\sigma(\lambda^i)_k - x_k)\tau(y)_k + (\sigma(\lambda^i)_k - x_k)\tau(y - e^k)_k\\
							& = & \sigma(\lambda^i)_k - (\sigma(\lambda^i)_k - x_k)\\
							& = & x_k,
	\end{eqnarray*}
	which means that the algorithm succeeded at computing a matching convex combination for $x$ at component $k$. It should be noted that the other components of $\lambda^{i+1}$ are unaffected by this update.
	
	
	Conversely, if $\lambda_y^i$ is less than the remaining difference between $\sigma(\lambda^i)_k$ and $x_k$, the point $y$ can be replaced completely by $y - e^k$. In this case the value of $\sigma(\lambda^{i + 1})_k$ remains greater than $x_k$
	
	\[\sigma(\lambda^{i + 1})_k = \sigma(\lambda^i)_k - \lambda^i_y \tau(y)_k + \lambda^i_y \tau(y - e^k)_k = \sigma(\lambda^i)_k - \lambda^i_y  > x_k\]
	Furthermore, the number of points in $\lambda^{i + 1}$ which have a value of $1$ at component $k$ is reduced by one with respect to $\lambda^i$. Considering that the number of points in $\lambda^i$ is finite, this implies that the algorithm must eventually compute a convex combination $\lambda''$ which matches $x$ at component $k$.


	To determine an upper bound on the number of iterations, it is helpful to observe that the size of the convex combination can only increase by $1$ for every iteration of the for loop, namely if $\lambda_y^i$ is greater than the difference between $\sigma(\lambda^i)_k$ and $x_k$. As a result, the number of points which comprise a convex combination during the $k$th iteration of the for loop is at most $\psi(\lambda') + k$. Since this number also gives an upper bound on the number of iterations performed by the while loop, the total number of iterations is at most
	
	\[\sum_{k = 1}^n (|\psi(\lambda')| + k) = n|\psi(\lambda')| + \sum_{k = 1}^n k = n|\psi(\lambda')| + \frac{n^2 + n}{2}.\]
	\qed
\end{proof}




\bibliographystyle{splncs}
\bibliography{literature}

\end{document}